\documentclass[sigconf, nonacm]{acmart}
\usepackage{amsthm}
\usepackage{graphics}
\usepackage{algorithm2e}
\RestyleAlgo{boxruled}

\usepackage{breakcites}

\usepackage{textcomp}
\usepackage{subcaption}
\usepackage{times}
\usepackage{soul}
\usepackage{graphicx}
\usepackage{bbding} 
\usepackage{xcolor}
\usepackage{color, colortbl}
\usepackage{enumitem}
\usepackage[english]{babel}
\usepackage{amsfonts}
\usepackage{helvet}
\usepackage{courier}
\usepackage{multirow}
\usepackage{booktabs}
\usepackage{mathtools}

\usepackage{placeins}
\usepackage{amsmath}
\newsavebox{\measurebox}
\newtheorem{rem}{Remark}

\newtheorem{proof}{Proof}
\newtheorem{LEM}{Lemma}[section]

\newcommand{\method}{\textsc{MDistrib}}
\newcommand{\methodmidas}{\textsc{MIDAS}}

\newcommand{\mycaption}{%
\ifx \@captype \@undefined \@latex@error {\noexpand \caption outside float}\@ehd \expandafter \@gobble \else \refstepcounter \@captype \expandafter \@firstofone \fi {\@dblarg {\@caption \@captype }}%
}%

\newcommand*\xor{\oplus}

\newcommand\vldbdoi{XX.XX/XXX.XX}

\newcommand\vldbavailabilityurl{https://github.com/prateekiiest/mdistrib}

\title{Distributed Anomaly Detection in Edge Streams using Frequency based Sketch Datastructures}

\author{
Prateek Chanda%
}
\email{prateek.dd2015@cs.iiests.ac.in}
\email{v-pratec@microsoft.com}
\affiliation{%
  \institution{Microsoft Research India}
  \streetaddress{Vigyan 1st floor, 9, Lavelle Road, Ashok Nagar,Bangalore}
  \city{India}
  \state{Karnataka}
}

\author{
Malay Bhattacharya
}
\email{malaybhattacharyya@isical.ac.in}
\affiliation{%
  \institution{Indian Statistical Institute}
  \streetaddress{Barrackpore Trunk Rd, Dunlop, Bonhooghly Government Colony, Baranagar}
  \city{India}
  \state{West Bengal}
}

\begin{document}

\begin{abstract}
Often logs hosted in large data centers represent network traffic data over a long period of time. For instance, such network traffic data logged via a TCP dump packet sniffer (as considered in the 1998 DARPA intrusion attack) included network packets being transmitted between computers. While an online framework is necessary for detecting any anomalous or suspicious network activities like denial of service attacks or unauthorized usage in real time, often such large data centers log data over long periods of time (e.g., TCP dump) and hence an offline framework is much more suitable in such scenarios. Given a network log history of edges from a dynamic graph, how can we assign anomaly scores to individual edges indicating  suspicious events with high accuracy using only constant memory and within limited time than state-of-the-art methods? We propose MDistrib and its variants which provides (a) faster detection of anomalous events via distributed processing with GPU support compared to other approaches, (b) better false positive guarantees than state of the art methods considering fixed space and (c) with collision aware based anomaly scoring for better accuracy results than state-of-the-art approaches. We describe experiments confirming that MDistrib is more efficient than prior work.
\end{abstract}

\maketitle


\section{Introduction}

Anomaly detection in dynamic complex graphs has a wide range of applications including intrusion detection, denial of service attacks and in the financial sector. Methods like \cite{su2019robust}, \cite{zheng2021generative}, \cite{he2020mtad} provide anomaly detection framework mostly in an offline scenario for timeseries multivariate graphs using recent advancements like Graph Neural Networks while some methods like \cite{odiathevar2019hybrid} provide a framework for a hybrid offline-online analysis for finding anomalies using Support Vector Machines (SVM).

Most large data centers include historical logs of network traffic data over long stretches of time like monthwise or daywise network activity. Such data might also be distributed across several nodes / partitions, for instance in a high performance computing cluster. Since often such distributed data is time series in nature, hence it is challenging to process and analyze them in parallel also since there might exist some inter-dependency across several nodes (see Fig.~\ref{fig:distributed nodes setting}).

\begin{figure}[h!]
    \centering
    \includegraphics[width=0.95\linewidth]{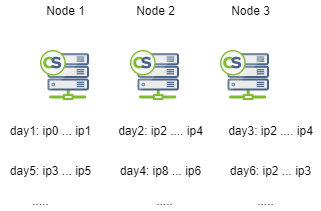}
    \caption{Network traffic data distributed across different nodes.}
    \label{fig:distributed nodes setting}
\end{figure}

While in an online setting, there has been several proposals in recent years that can detect anomalous edges with a high accuracy like MIDAS \cite{bhatia2020midas}, Spotlight
 \cite{eswaran2018spotlight} , there has not been any rigorous work done for an offline setting phase. While \cite{bhatia2020midas} performs better than all the recent state of the art methods and can even process edges in the order of millions in seconds. But in the case of large network logs that are hosted over months time, and which can consist of billions of edges, methods like MIDAS still take considerable time to execute. (For example, for 1 billion edges, MIDAS takes 250 seconds). 

In this paper, we thus propose a distributed version of anomalous edge detection for large time series graphs using constant memory space called $\method$ as an extension to previous work on $\methodmidas$ \cite{bhatia2020midas}, thereby reducing overall edge processing time and faster execution. In addition, we extend our method for any generalized frequency based sketch data structure with better theoretical guarantees for frequency estimates and better theoretical bounds on false positive probability as compared to MIDAS. 

Our main contributions in the paper are as follows:
\begin{enumerate}
\item\textbf{Distributed Framework:} We propose a distributed framework for anomaly detection in large graphs with better accuracy and performance than state-of-the-art approaches.
\item \textbf{Constant Memory Space:} The corresponding framework provides a constant space solution to anomaly detection via using sketch based data structures. 
\item \textbf{Better Theoretical Guarantees:} In subsequent theorems we showcase how the false positive probability is more tightly bound in case of $\method$ as compared to $\methodmidas$, while other methods in Table~\ref{tab:comparison} does not showcase such theoretical bounds.

\item \textbf{Extension to Generalized Frequency based sketches} We describe our model based on a generalized frequency sketch and then showcase it via Count-Min-Sketch : $\method_{CMS}$ and Apache Frequent Item Sketch : $\method_{FIS}$.

\item \textbf{GPU enabled anomaly detection} In addition to the distributed framework, we also extend $\method$ to support GPU and showcase via further experiments how $\method$ performs better in a GPU enabled environment.

\item \textbf{Effectiveness} Our experimental results also showcases that $\method$ and its variants show a high performance in terms of execution (382 - 815 times faster than other baseline algorithms ) and a high accuracy comparable or better (42\% - 48\%) than state of the art methods.

\end{enumerate}

\begin{table}[!ht]
\centering
\caption{Comparison of relevant methods}
\label{tab:comparison}
\noindent\setlength\tabcolsep{2pt}%
\begin{tabular}{@{}rcccccc|c@{}}
\toprule
 & \rotatebox{90}{ SedanSpot '18} 
 & \rotatebox{90}{Spotlight}
 & \rotatebox{90}{PENminer}
 & \rotatebox{90}{F-Fade}
 & \rotatebox{90}{MIDAS-R}
 & {\bf \rotatebox{90}{\method}} \\ 
 \midrule
\textbf{Microcluster Detection} & - & - & -  & - & - & \Checkmark & \CheckmarkBold \\
\textbf{Constant Memory} & \Checkmark & \Checkmark & \Checkmark & - & \Checkmark & \Checkmark  & \CheckmarkBold \\
\textbf{Constant Update Time} & \Checkmark & \Checkmark & \Checkmark & \Checkmark &  \Checkmark &  \Checkmark &  \CheckmarkBold \\
\textbf{Distributed Processing} & - & - & - & - & - & - & \CheckmarkBold \\
\textbf{GPU support} & - & - & - & - & - & - & \CheckmarkBold \\
\bottomrule
\end{tabular}
\end{table}

\section{Related Work}
There have been substantial amount of works done in the area of anomaly detection from static network data \cite{chandola2009anomaly, bhuyan2013network}, including the current focus on deep learning based predictions \cite{pang2021deep}. The recent deep learning approaches toward anomaly detection have shown immense promise in learning feature representations or anomaly scores via neural networks. Many of these models, including the GAN-based models like EBGAN \cite{zenati2018efficient}, fast AnoGAN \cite{schlegl2019f} and BiGAN \cite{donahue2016adversarial}, demonstrate a significantly better performance than conventional approaches in diverse real-life applications. However, there exist some limitations that are encountered by deep learning methods \cite{pang2021deep}. The most important of this is dealing with high-dimensional data. Anomaly detection from large-scale data
is challenging with deep learning approaches. Moreover, given the fact that, as there are limited amount of labeled anomaly data, it remains a challenge to learn expressive normality/abnormality representations from the data. Often in an offline setting, models like neural networks may provide better accuracy but take a huge amount of time to compute.

The recent interest in anomaly detection from data streams has also contributed quite a few studies in the literature \cite{tam2019anomaly}. The studies on anomaly detection from streaming graphs is relatively newer. Prior to this, when anomaly detection was confined to only static data, there were three types of variants of this problem. This was based on the anomalous entities appearing in the graph, be it nodes, edges or subgraphs. However, based on the way the graph is received in a streaming setting, additional problems emerged in due course of time.

A recent study by Liu et al. performs streaming anomaly detection with frequency and patterns \cite{liu2021isconna}. Their work proposes online algorithms (termed as Isconna-EO and Isconna-EN) that measure the consecutive presence and absence of edge records. Though the variants of Isconna demonstrates a comparatively better accuracy than the state-of-the-art, however it suffers from parameter tuning. Very recently, Chang et al. have proposed a frequency factorization approach (termed as F-fade) for anomaly detection in edge streams \cite{chang2021f}.

The novelty that we introduce in the current paper is providing a fast and scalable parallel framework for distributed sketch based anomaly detection in edge streams in GPU environment with collision aware optimization.

\section{Problem Setting}
Let there be a dynamic time-evolving graph $G$ characterized by the stream of edges $\varepsilon = (e_1, e_2, \dots)$, $e_i$ denoting $(u_i, v_i, t_i)$. Here, $u_i$ and $v_i$ symbolize source and destination nodes respectively and $e_i$ denotes the incoming edge between them in timestamp $t_i$. Given $G$, the concern is how can we detect anomalous edges and micro-cluster activities in a distributed fashion. In a network setting, $u_i$ and $v_i$ may indicate source and destination IP addresses respectively and edge denoting a connection rpapequest occurring at timestamp $t_i$. Also the graph $G$ can be modelled as a directed multi-graph indicating that $(u_i,v_i) \neq (v_i, u_i)$ and edges might appear more than once between any two specific nodes.

\subsection{Temporal Scoring}

We first present here the temporal scoring approach proposed by our baseline method \cite{bhatia2020midas}.

\begin{figure}[h!]
    \centering
    \includegraphics[width=0.8\linewidth]{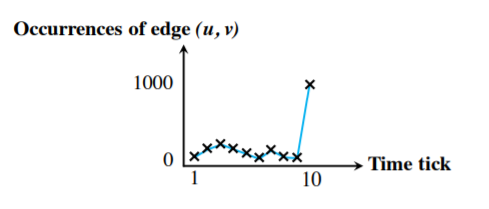}
    \caption{ Time series of a single source-destination pair
(u, v), with a large burst of activity at time tick 10.}
    \label{fig:midasMotivating}
\end{figure}

We start by proposing a motivating example via Figure \ref{fig:midasMotivating} on how MIDAS computes the anomaly score for each incoming edges. As per Figure \ref{fig:midasMotivating}, considering a single source destination pair ($u$,$v$), there is a sudden occurrence of large burst of activity at time instant t=10. So the main motivation is to detect the immediate displacement in the item frequency from its corresponding mean. Hence a  higher the displacement indicates a higher chances of being an actual anomaly. Taking forward this intuition, \cite{bhatia2020midas} maintains two types of Count-Min Sketch (CMS) \cite{cormode2005improved} data structures. Considering $s_{uv}$ as the total number of edges from node $u$ and $v$ up to current time tick $t$, this can be stored in a single CMS called Current Sum Sketch (CS) and likewise all $s_{uv}$ for all possible combinations of nodes will be stored.

Similarly for the current number of edges between node $u$ and $v$ existing during the current time tick $t$ can be denoted as $a_{uv}$, which can be again similarly tracked in a separate CMS called Current Edge Sketch (CE). The major difference between the two sketches lies only in the fact that after every time tick $t$, the values in CE needs to be resetted back to 0 in order to accommodate new information for the next timestamp.  

So now, given the estimated frequencies $\hat{s_{uv}}$ and $\hat{a_{uv}}$, how can we assign corresponding anomaly scores to the edge $u-v$? MIDAS proposes to use chi-squared goodness-of-fit test as the hypothesis testing framework to compute such anomaly scores. A major reason to go forward using chi-squared test is due to avoid the assumption that the edge frequencies follow some distribution and that using that underlying distributing information some likelihood estimate for the edge indicating anomaly can be formulated.

The past edges can be divided into two classes: current time tick (t=10) i.e. $a_{uv}$ and all past time ticks (t<10) i.e $s_{uv} - a{uv}$.

Hence, the chi-squared statistic which is defined as sum over categories of 
$\frac{(observed - expected)^2}{expected}$ can be further written as

\begin{equation}
\begin{aligned}
X^2 = \frac{({observed}_{curr_t} - {expected}_{curr_t})^2}{{expected_{curr_t}}} \\ 
+ \frac{({observed}_{past_t} - {expected}_{past_t})^2}{{expected_{past_t}}} \\
= (~{a_{uv}} -  \frac{~{s_{uv}}}{t})^2 \frac{t^2}{~s_{uv}(t-1)}
\end{aligned}    
\label{chi-square : derivation}
\end{equation}

\begin{rem}
We have avoided the derivation of Equation \ref{chi-square : derivation} for simplicity due to the same derivation structure in \textit{Hypothesis testing framework} in \cite{bhatia2020midas}. 
\end{rem}

Therefore, given an arriving edge indicated by tuple $(u,v,t)$, our anomaly score is computed as:

 \[  score((u,v,t)) = (\hat{a_{uv}} -  \frac{\hat{s_{uv}}}{t})^2 \frac{t^2}{\hat{s_{uv}}(t-1)}    \]

\section{Proposed Distributed Microcluster Detection}

\begin{figure}[h!]
\centering
\includegraphics[width=1  \columnwidth]{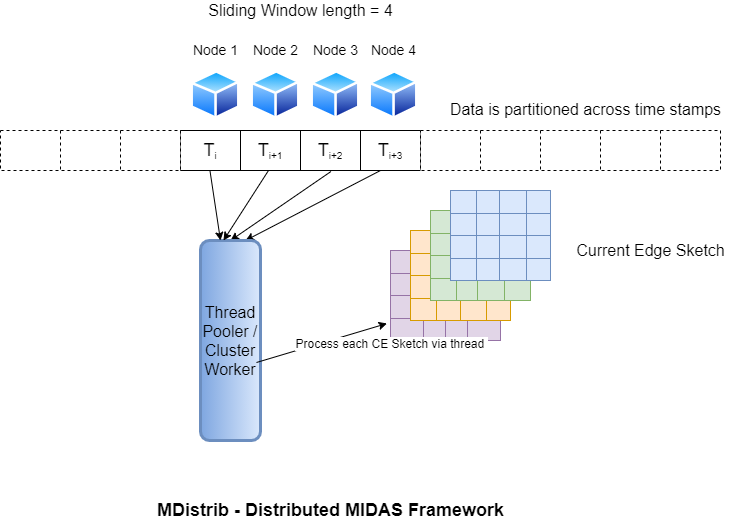}
\caption{\method
Architecture}
\label{fig:midasDsArch}
\end{figure}

In the offline phase, we may consider a scenario where there exists network logs consisting of incoming requests during the past month and the data is partitioned across nodes/clusters by its corresponding timestamp (daywise).

Instead of a linear scan, our proposed framework models the data stream $S = (d_1, d_2, d_3, \dots d_t, \dots)$ s.t. $d_k = (i_k, t_k)$, where $i_k$ indicates the corresponding item and $t_k$ denotes the timestamp. Considering there are $n$ timestamps, for each single timestamp there might be multiple items which occur. Each of the timestamps can be considered for parallel processing rather than serially processing. Considering there are $m+1$
hardware threads in a multi-processor, the stream can be partitioned into m sub-streams $S_1, S_2,S_3, \dots S_m$.

For each substream, we consider only one frequency based sketch data structure denoted as $CE$. This is employed for maintaining the current frequency per item. And at the end of $m$ sub-streams we maintain another sketch denoted as $CS_t$ which stores the sum total frequency corresponding to each edge uptil all such previous substreams.
Each such $CE$ sketch data structure indexed by timestamp is stored under a particular thread Id.

\begin{figure}[h!]
\centering
\includegraphics[height=0.4\columnwidth, width=0.98\columnwidth]{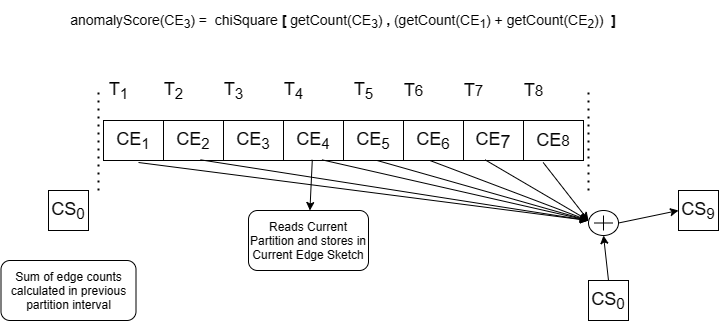}

\caption{Calculation of CE and CS sketches.}
\label{fig: Collision CMS}
\end{figure}


Since $CS$ stores the total sum frequency for a particular source destination pair ($u_i, v_i$), for a particular time instant $T_i$, it can be written as follows $CS(T_i) =\sum_{j=0}^{j=i-1}CE(T_j)$ i.e summation over all previous \textit{current Edge} frequency sketches.

\textbf{Algorithm}
At the beginning, we propose the formal algorithm for the offline phase of our proposed \method. Note that, sketch data structures mentioned here are general frequency based sketch data structures that share mergeability property. In order to be useful in large distributed computing environments, the sketches must be mergeable without additional loss of accuracy. This can be formally stated as $sketch(A + B) \approx sketch(A) \cup sketch(B)$.

\begin{algorithm}
\KwIn{Stream of graph edges distributed across $K$ partitions}
\KwOut{Anomaly score of individual edges}
$\text{Initialize } m \text{ Current Edge(CE) Sketches}$ \\
$\text{Initialize } m \text{ Current Sum(CS) Sketches}$ \\ 
\While{$\text{All K partitions are visited}$}{
            $\text{Assign worker per m  partition}$\\ 
            \For{$m$ partitions in parallel}{                 
                $\text{Execute }  CalculateCE(partition^{m}_{i}) \text{ in parallel}$
            }
            \For{$m$ partitions in parallel}{                
                $\text{Execute }   CalculateScore(partition^{m}_{i}) \text{ in parallel}$
            }
            }
\caption{Distributed Microcluster Detection - Offline}
\end{algorithm}

\begin{algorithm}
\SetAlgoLined
\caption{CalculateCE() -- Current Edge Calculation}
\KwIn{Stream of graph edges for Partition id $\theta$ }
\While{$\text{a new edge } (u,v,t_\theta) \text{ arrives}$}{
    $\text{Assign edge } (u,v) \text{ to a thread}$
    $\text{Update } CE^m_{i} \text{ data structures for the new edge } uv$
}
\end{algorithm}

\begin{algorithm}
\caption{CalculateScore() -- Current Edge Score Calculation}
\KwIn{Stream of graph edges for Partition id $\theta$}
\KwOut{Anomaly Scores per edge}
\While{$\text{a new edge } (u,v,t_\theta) \text{ arrives}$}{
     $\text{Assign edge } (u,v) \text{ to a thread}$
     $\text{Perform merge operation for } CS^m_{i} = \sum{CE^m_{i}}$
      $\text{Retrieve updated counts } \hat{a_{uv}} \text{ from } CE^m_{i}$
     $\text{Retrieve updated counts } \hat{s_{uv}} \text{ from } CS^{M}_{i}$\\
    $\textbf{  Compute Edge Scores }$\\
    $score((u,v,t)) = (\hat{a_{uv}} -  \frac{\hat{s_{uv}}}{t})^2 \frac{t^2}{\hat{s_{uv}}(t-1)}$\\
    \textbf{Return} $score$\\
}
\end{algorithm}

\subsection{\method-CMS $\xrightarrow{}$  Distributed MIDAS using Count Min Sketch}
Here, we follow the exact architecture of $\method$  as proposed above and use count min sketch for each of the $m$ \textit{Current Edge Sketch} and \textit{Current Sum Sketch}, considering $m$ partitions, with each having $w$ and $b$ number of hash functions and number of buckets respectively.

\begin{figure}[h!]
    \centering
    \includegraphics[width=0.7\linewidth]{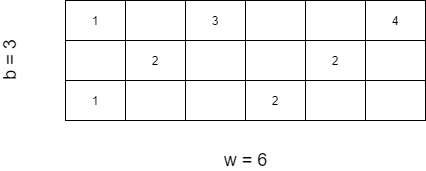}
    \caption{CMS with parameters 6x3}
    \label{fig:cms Basic}
\end{figure}

For example Figure \ref{fig:cms Basic} indicates the basic functionality in a count min sketch where each item can be hashed into the corresponding slots.

\subsection{\method-FIS - Distributed MIDAS using Apache Frequent Item sketch}
 \cite{anderson2017high} introduces Frequent Item Sketch which belongs to a class of Apache Incubator Datasketches that deals with estimating frequency of items in a streaming environment and is mainly aimed at applications like Heavy-Hitter estimation or estimation of counts of objects in a streaming environment. The Apache Incubator Datasketches framework provide a wrapper api for Frequent Item sketch with the mergeability functionality where two Frequent Item Datasketches $f_1$ and $f_2$ can be merged into one frequent item datasketch $f_r$ with all the frequencies remaining intact $f_r = f_1 \xor f_2$.


\subsection{$\method R$ :  Distributed MIDAS-R}
Similar to $\method$, we take into account the idea of relation attribute used in \cite{bhatia2020midas} for both temporal and spatial relations.

\textbf{Temporal Relations}
With the idea of most recent edges being a better representative of the statistic score than previous early edges, for each timestamp id $Tid$ within the partition interval we apply a decay factor $\alpha^{\frac{1}{Tid}}$ where $\alpha \in (0,1)$. This is done keeping into consideration the fact that edges from the recent past should compute to the anomaly score calculation as well.

\textbf{Spatial Relations}

As it is also important to capture large group of spatially nearby edges: such as a high out-degree connection from a single source IP address, following MIDAS's approach we also consider observing \textbf{nodes} with high degree as they occur. For this, we create CMS counters $\hat{a_{u}}$ and $\hat{s_{u}}$ to approximate the number of edges incident on node $u$ at the current time tick $t$ and up till $t$ respectively. Hence, for a new incoming edge $u,v$, we compute the anomaly score for edge $u,v$ as well as for node $u$ and node $v$ and finally take the maximum of the three scores.




\subsection{Time complexity Analysis}
We focus on time complexity analysis and comparison with respect to $\methodmidas$ since $\methodmidas$ has performed significantly well in comparison with SOTA models in terms of time complexity.

Considering we have $K$ partitions, $P_{j+1}, P_{j+2}, \dots P_{j+k}$, where each partition $P_{j+i}$; $\forall{i \in [1,K]}$ belonging to timestamp $T_{j+i}$ includes a $CE$ datasketch. We denote the $CE$ sketch update time as $T_{CE}$. If the average number of items during $T_{j+i}$ is $\mu$, then $T_{CE} \approx \mu$ and the total update time for all K partitions under the MIDAS framework would be $\mu*T_{CE}$. Similarly for calculating the sum of the previous edges seen uptil time $T_{j+i}$ it can be denoted as $T_{CS} = \tau_c * (K-1)$ where $\tau$ indicates the merge time for two datasketches, Since in both \method and \methodmidas there is merging of previous current edge sketches, hence $T_{CS}$ remains the same for both.

Comparing the total computation for both the methods, yields

\[T_{\methodmidas} = \phi[\mu*T_{CE} + \tau_c * (K-1)] \]

\[T_{\method} = \phi[\mu + \tau_c * (K-1) + Th_{AllocationTime}],\]
s.t. $\phi*K \approx C$ where $C$ denotes the number of total partitions/nodes.

Considering thread allocation time $Th_{AllocationTime}$ to be very negligible compared to other factors like large number of streaming data points per partition $\mu$ and large number of partitions $K$, it can be seen  that $T_{\methodmidas}$ is much less than $T_{\method}$.

Further we also propose reading items per partition to be distributed across various subthreads instead of sequential reading, since each edge under a particular partition/ timestamp can be processed independently of the other. As a result, update time for $CE$ denoted by $\mu$ can be further reduced to $\hat{\mu}$ where $\hat{\mu} < \mu$ thus resulting in  much faster execution time. 



\subsection{Better theoretical guarantees}

\textbf{Count Min Sketch Error Estimate}

\begin{LEM}
The error estimate encountered via querying edge counts in $\method$ is less than that of $\methodmidas$ by a factor of ($\sum_xCount - \max(Count)$)
\label{mdsitrib-cs-error-gurantee}
\end{LEM}

\begin{proof}
$\method$CMS has the same sketch architecture as $\methodmidas$, but however presents with a better theoretical guarantee as shown below:

As per Therrem~1 \cite{cormode2005improved}, the estimate $\hat{a}$ from a CMS has the following guarantees: $a \leq \hat{a}$ with some probability $1 - \frac{\delta}{2}$,

\begin{equation}
\begin{aligned}
\hat{a} \leq a + \epsilon*V
\label{eq:CMSBounds}
\end{aligned}
\end{equation}

where the values of these error bounds can be chosen by $w = \lceil{\frac{e}{\epsilon}}\rceil$ and $d = \lceil\ln(\frac{1}{\delta})\rceil$, $w$ and $d$ being the dimensions of the CMS. $V$ indicates the total sum of all the counts stored in the data structure.
Again considering $K$ partitions $P_{j+1}, \dots P_{j+k}$ , let $P_{j+\theta}$ where $\theta \in [1, K]$ denote the partition $\theta$. As per $\methodmidas$ the sum sketch for time tick $t$ is computed by summation of all previous current edge sketches $CE$ uptil $t-1$. Hence for partition $\theta$,  $V_{\theta}^{MIDAS}$ can be written as follows.

\[V_{\theta}^{MIDAS} = \sum_{r=1}^{r=\theta-1}V_{r}^{MIDAS} \]

\[V_{\theta}^{MIDAS} = \sum_{r=1}^{r=\theta}\sum_{p=1,q=1}^{w,d}\text{getCount}_{r}(p,q)\]

For \method-$CMS$, however since for the partition $\theta$ we have $\theta - 1$ CE sketches, instead of merging them for computing the same, we can individually query each of the $\theta-1$ sketches instead of merging them. Hence in this case,
\[V_{\theta}^{\method} = max(V_{r}^{\method})   \text{ where } r \in [1, \theta -1 ]\]

Hence, as per Equation.~(\ref{eq:CMSBounds}) we compare the estimates of the count retrieved from the current Sum $CS$ CMS sketch of partition $\theta$ for $\method$ and $\methodmidas$.

\begin{equation}
\text{MIDAS Estimate} : \hat{a_{\theta}} \leq a + \epsilon*(\sum_{r=1}^{r=\theta-1}\text{count}_{r})
    \label{eq:MIDAS Estimate}
\end{equation}

\begin{equation}
   \text{MDistrib Estimate} :\hat{a_{\theta}} \leq a_{\theta} + \epsilon*\max(\text{count}_{1}, \dots \text{count}_{\theta -1 })
   \label{eq:MDistrib Estimate}
\end{equation}

Since such counts stored in the CMS being non-negative

\begin{equation}
\max(\text{count}_{1}, \dots \text{count}_{\theta -1 }) <= \sum_{r=1}^{r=\theta-1}\text{count}_{r}
\label{eq:positiveInequality}
\end{equation}

Hence, taking the inequality from \ref{eq:positiveInequality} we can conclude that the error bound in the estimate provided by $\method$ is much less than that provided by $\methodmidas$.

\begin{equation}
 \begin{aligned}
 \hat{a_{\theta}} \leq a_{\theta} + \epsilon*\max(\text{count}_{1}, \dots \text{count}_{\theta -1 }) \\
 \leq a_{\theta} + \epsilon*\sum_{r=1}^{r=\theta-1}\text{count}_{r}
 \label{eq: better guarantee CMS}
\end{aligned}
\end{equation}
\qed
\end{proof}

\textbf{Apache Frequent Item Sketch Guarantees}
In case of a Frequent Item Sketch with configuration $<M, N>$, where $M$ indicates the maximum map size in power of 2 and $N$ being the total weights stored in the map,  the error in the \textit{estimate} of the frequency of a particular edge is always bounded within some $T$, $\hat{a} \leq a + T$, where $T$ can be as follows, given $\alpha$ is the load factor :

\begin{equation}
    T=
    \begin{cases}
      0, & \text{if}\ N < \alpha * M \\
      \epsilon*N, & \text{otherwise}
    \end{cases}
  \end{equation}

Since Frequent Item Sketch has a similar structure as a Count Min Sketch in terms of the error estimate, the proof for the false positive probability bound also holds here. 

\begin{LEM}
The error estimate encountered via merging two sketches $A$ , $B$ is the same as that obtained in an original sketch $C$ which contains frequencies all items that have been inserted to $A$ and $B$ as a whole.
\end{LEM}

\textbf{Establishing for Count Min Sketch}

Considering count min sketches $A$ and $B$ having dimensions $w \times d$ each and inserting an item with actual frequency $a$ in both the sketches, then the error estimate of a particular item can be represented as follows

\begin{equation}
\begin{aligned}
\hat{a_A} \leq a + \epsilon*V_{A}
\label{eq:CMSBounds for A}
\end{aligned}
\end{equation}

\begin{equation}
\begin{aligned}
\hat{a_B} \leq a + \epsilon*V_{B}
\label{eq:CMSBounds for B}
\end{aligned}
\end{equation}

where $V_A$ and $V_B$ indicate the total sum of all the counts stored in CMS $A$ and $B$ respectively.

By combining Equations.~(\ref{eq:CMSBounds for A}) and (\ref{eq:CMSBounds for B}), and since all the elements involved are ($\geq 0$), we get the following.

\begin{subequations}
 \label{eq1Merge}
 \begin{align}
  \hat{a_A} + \hat{a_B}\leq 2a + \epsilon*(V_{A} +V_{B}) \label{eq11} \\
  \hat{a_{A+B}} \leq 2a + \epsilon*(V_{A} +V_{B}) \label{eq12} \\
   \hat{a_{C}} \leq 2a + \epsilon*(V_{C}) \label{eq13}
 \end{align}
\end{subequations}

where $\hat{a_A} + \hat{a_B} = \hat{a_{A+B}}$ indicates the total estimated frequency of that particular item, had that item be inserted into the sketch $C$ having same dimension $wxd$ which is a combination of $A$ and $B$. Hence by this we showcase that the error estimate remains the same even after merging two Count Min Sketches compared to an original Count Min sketch which represents a summation of the two other sketches.

\subsection{Why Distributed MIDAS using Apache Frequent Item sketch is better than MIDAS using CMS}

Frequent Item Sketch only stores those items at any point that's larger in freq than the global median frequency, whereas CMS keeps on storing Estimation Errors. Dimitropoulos et al. have introduced the concept of transient keys, which are items that are not active and have low frequency relative to the group, much earlier \cite{dimitropoulos2008eternal}.


\subsection{Space Complexity}

\textbf{Constant Space for \method-CMS}

Considering $K$ distributed partitions in our system, at any point we have only $K$ current Edge $CE$ CMS sketches. After the end of $K$ partitions, again when we consider the next $K$ partitions, in order to store the previous iteration's sum we maintain another CMS.
Hence the total space complexity $(K+1)*(wd)$ where $w$ and $d$ indicate the number of hash functions and buckets in the CMS structure. 

\textbf{Constant Space for \method-FIS}

Similarly for Frequent Item Sketch, the total space complexity can be denoted as $(K+1)*(M)$ where $M$ is the maximum map size of the data sketch.

\subsection{Proofs}

\textbf{Proposition 1 : False Positive Probability Error is more tightly bound in case of \method}

\vspace{2pt}

\textit{Count Min Sketch}

As we know for a particular partition $\theta$, our chi-square statistic can be written as

\[ \hat{\chi_{\theta}^2}  = (\hat{a^{\theta}_{uv}} -  \frac{\hat{s^{\theta}_{uv}}}{t})^2 \frac{t^2}{\hat{s^{\theta}_{uv}}(t-1)}    \]
  
\textbf{Theorem 1 (False Positive Probability Bound)}

Let $\chi^{2}_{1 - \frac{\epsilon}{2}}(1)$ be the $(1-\frac{\epsilon}{2})$-quantile of a chi-squared random variable with 1 degree of freedom. Then:

\begin{equation}
\begin{aligned}
P(\hat{\chi^{2}} > \chi^{2}_{1 - \frac{\delta}{2}}(1)) < \delta 
\label{eq: chistatisticQuantile}
\end{aligned}    
\end{equation}

\cite{bhatia2020midas} mentions that the motivation of the above theorem proposal is to show that with a threshold of $\chi^{2}_{1 - \frac{\epsilon}{2}}(1)$,  $\chi^{2}$ as a test statistic results in a false positive probability of at most $\delta$.

\vspace{3pt}

For CMS guarantees \ref{eq:CMSBounds} we have the following

\begin{equation}
\begin{aligned}
P(\hat{a} \leq a + \tau) \geq 1 - \frac{\delta}{2}
\label{eq :cmsGuaranteeBound}
\end{aligned}    
\end{equation}

where $\epsilon*V$ is shortened as $\tau$

We define an adjusted version of our earlier score for the total summation score :

\begin{equation}
    \Tilde{s^{\theta}_{uv}} = \hat{s^{\theta}_{uv}} - \tau
    \label{eq:updatedEquationSum}
\end{equation}

By combining Equations.~\ref{eq: chistatisticQuantile} and \ref{eq :cmsGuaranteeBound} by union bound, with probability of at least $1-\delta$, we obtain the following.

\[
 \Tilde{\chi_{\theta}^2}  = (\hat{a^{\theta}_{uv}} -  \frac{\Tilde{s^{\theta}_{uv}}}{t})^2 \frac{t^2}{\Tilde{s^{\theta}_{uv}}(t-1)}    \\
 \]
 
We denote $\Tilde{\chi_{\theta_{MIDAS}}^2}$ and $\Tilde{\chi_{\theta_{MDISTRIB}}^2}$ as the chi-square test statistic for  $\methodmidas$ and $\method$ methodologies respectively.

 \begin{subequations}
 \label{eq1ChiMIDAS}
 \begin{align}
 \Tilde{\chi_{\theta_{MIDAS}}^2}  = (\hat{a^{\theta}_{uv}} -  \frac{\Tilde{s^{\theta}_{uv_{MIDAS}}}}{t})^2 \frac{t^2}{\Tilde{s^{\theta}_{uv_{MIDAS}}}(t-1)}  \label{eqChiMIDAS} \\
 \Tilde{\chi_{\theta_{MDISTRIB}}^2}  = (\hat{a^{\theta}_{uv}} -  \frac{\Tilde{s^{\theta}_{uv_{MDISTRIB}}}}{t})^2 \frac{t^2}{\Tilde{s^{\theta}_{uv_{MDISTRIB}}}(t-1)}  \label{eqChiMDistrib} \\
 \end{align}
\end{subequations}

We simplify the equation for $\method$ as it holds the same for $\methodmidas$. From Equation.~(\ref{eq:updatedEquationSum}), Equation.~(\ref{eqChiMDistrib}) can be simplified as follows.

\begin{equation}
\begin{aligned}
    \Tilde{\chi_{\theta_{MDISTRIB}}^2}  = \\ (\hat{a^{\theta}_{uv}} -  \frac{\hat{s^{\theta}_{uv_{MDISTRIB}}} - \tau}{t})^2  \frac{t^2}{\Tilde{s^{\theta}_{uv_{MDISTRIB}}}(t-1)} 
\end{aligned}    
\end{equation}

Considering $a^{\theta}_{uv} \approx \gamma\frac{s^{\theta}_{uv}}{t}$ under the assumption that the current edge count is always an estimation of the previous mean plus some positive weight factor,

\begin{equation}
\begin{aligned}
\Tilde{\chi_{\theta_{MDISTRIB}}^2}  = \\ ( \gamma\frac{\hat{s^{\theta}_{uv_{MDISTRIB}}}}{t} -  \frac{\hat{s^{\theta}_{uv_{MDISTRIB}}}  - \tau}{t})^2 \frac{t^2}{\Tilde{s^{\theta}_{uv_{MDISTRIB}}}(t-1)}     
\end{aligned}    
\end{equation}

which can be simplified as considering large value of $t$, that is for higher timestamps.

\[
\Tilde{\chi_{\theta_{MDISTRIB}}^2}  \approx \hat{s^{\theta}_{uv_{MDISTRIB}}}(\frac{\gamma -1}{t})\frac{t^2}{(t-1)}
\]

\begin{equation}
    \Tilde{\chi_{\theta_{MDISTRIB}}^2}  \approx \hat{s^{\theta}_{uv_{MDISTRIB}}}(\gamma -1)
\end{equation}

\begin{equation}
    \Tilde{\chi_{\theta_{MIDAS}}^2}  \approx \hat{s^{\theta}_{uv_{MIDAS}}}(\gamma -1)
\end{equation}

As per Lemma~\ref{mdsitrib-cs-error-gurantee}, we know that ${\hat{s^{\theta}_{uv_{MDISTRIB}}}} \leq {\hat{s^{\theta}_{uv_{MIDAS}}}}$

Hence, by combining both the equations, we get $\Tilde{\chi_{\theta_{MDISTRIB}}^2} < \Tilde{\chi_{\theta_{MIDAS}}^2}$




\section{Experiments}
In this section, we evaluate our proposed anomaly scoring algorithm \method\ and address to answer the following questions:

\begin{enumerate}[label=\textbf{Q\arabic*.}]
\item {\bf Accuracy:} How accurately does \method\ and its variants detect real-world anomalies compared to baselines, as evaluated using the ground truth labels?

\item {\bf Scalability:} How does it scale with input stream length? How does the time needed to process each input compare to baseline approaches?

\item {\bf Execution Time:} How fast does the proposed method perform with respect to baseline methods while generating anomaly scores in the offline phase?

\end{enumerate}

\paragraph{Baselines: MIDAS-R, SedanSpot,  PENnimer, F-Fade}

\paragraph{Datasets}

In this section, we evaluate performance of $\method$ as compared to other baselines in Table \ref{Table : AUC} based on 5 real world datasets :
\emph{DARPA} \cite{lippmann1999results} and \emph{ISCX-IDS2012} are popular datasets for graph anomaly detection. According to survey [52], there has been proposal of more than 30 datasets and it has been recommended to use \emph{CIC-DDoS2019} and \emph{CIC-IDS2018} datasets. We also showcase the same for \emph{CTU-13} dataset. 

\begin{table}[h!]
\centering
\resizebox{0.7\columnwidth}{!}{%
\begin{tabular}{l|c|c|c}
\hline
Datasets & |V| & |E| & |T| \\
\hline
DARPA & 25,525 & 4,554,344 &  46,567 \\
ISCX-IDS2012 & 30,917 & 1,097,070 & 165,043 \\
CIC-IDS2018 & 33,176 & 7,948,748 &  38,478 \\
CIC-DDoS2019 & 1,290 & 20,364,525 & 12,224 \\
CTU-13  & 371,100 & 2,521,286 & 33,090 \\
\end{tabular}%
}
\caption{Statistics of the datasets.}
\label{Table : data statistics}
\end{table}

Table \ref{Table : data statistics} shows the statistical summary of the datasets. |E| corresponds to the total number of edge
records, |V| and |T| are the number of unique nodes and unique timestamps, respectively.

\subsubsection{Hyperparameters chosen for baselines}

We hereby mention the hyperparameters used by the corresponding baselines.

\vspace{2pt}

\textbf{Spotlight}

\begin{itemize}
\item \text{sample size} = 10000
\item \text{num walk} = 50
\item \text{restart prob} =  0.15
\end{itemize}

\vspace{2pt}

\textbf{MIDAS}
The size of CMSs is 2 rows by 1024 columns for all the tests. For MIDAS-R, the decay factor $\alpha$ = 0.6.

\textbf{PENimer}
\begin{itemize}
    \item ws =1 
    \item ms =1
    \item alpha = 1
    \item beta =1
    \item gamma =1
\end{itemize}

\vspace{2pt}

\textbf{F-Fade}
\begin{itemize}
    \item embedding size = 20
    \item $W_{upd}$ = 720
    \item $T_{th}$ = 120
    \item alpha = 0.999
    \item M =100
\end{itemize}

For $t_{setup}$, we always use the timestamp value at the 10th percentile of the dataset.

\vspace{2pt}

All methods output an anomaly score corresponding to the incoming edge, a high anomaly score implying more anomalousness. As part of the experiments, we follow the same suite as the baseline papers in reporting the Area under the ROC Curve (using the True Positive Rate TPR and the False Positive Rate FPR) and the corresponding running time.  All experiments for calculating the AUC are averaged for 8 runs and the mean is reported along with its spread.

\subsection{Experimental Setup}
All experiments are carried out on a $1.00 GHz$ Intel Core $i5$ processor, $8 GB$ RAM, running OS $Win10$ $10.0.19042$. We used an open-sourced implementation of MIDAS, provided by the authors, following parameter settings as suggested in the original paper ($2$ hash functions, $719$ buckets). We implement $\method$ in python using dask framework. We follow the same parameter settings as $\methodmidas$ ($2$ hash functions, $719$ buckets).

\subsection{Q1. Accuracy}

Table \ref{Table : AUC} includes the AUC of edge anomaly detection baselines along with our two proposed methods $\method_{CMS}-R$ and $\method_{FIS}-R$.
PENimer is unable to finish in large datasets like CIC-DDos and CIC-IDS within 10 hours range, thus we do not report the result. Among the baseline results $\method_{FIS}-R$ and $\method_{CMS}-R$ both provide better accuracy results than the baselines on datasets like DARPA, CIC-DDoS and ISCX-IDS.

\textbf{AUC vs Running Time}
Figure \ref{fig:rocvstime}. plots accuracy (AUC) vs. running time (log scale, in seconds) on the Darpa dataset. As we can see both $\method_{FIS}-R$ and $\method_{CMS}-R$ achieve higher accuracy results than the base lines in considerable less time.

\begin{figure}[h!]
    \centering
    \includegraphics[width=0.98\linewidth, height=0.35\columnwidth]{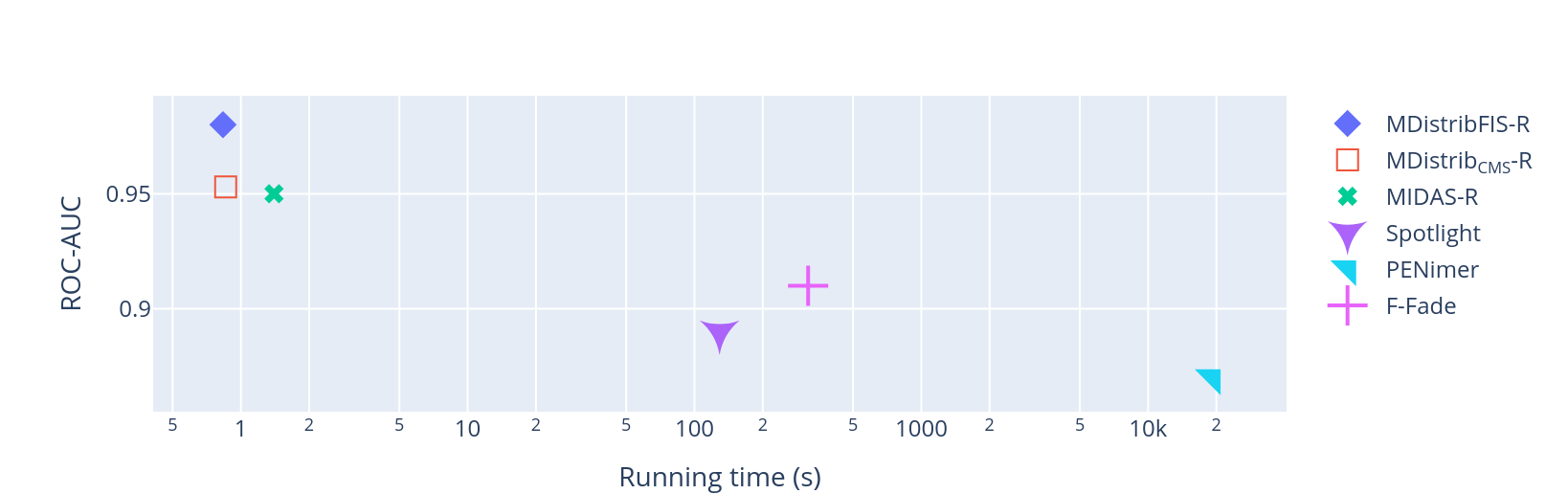}
    \caption{AUC vs running time when detecting edge anomalies from the DARPA dataset.}
    \label{fig:rocvstime}
\end{figure}

\begin{table*}[h!]
\centering
\resizebox{1.4\columnwidth}{!}{%
\begin{tabular}{|l|c|c|c|c|c|}
\hline
Methods & \textbf{DARPA} & \textbf{CIC-DDoS} & \textbf{CIC-IDS} & \textbf{CTU-13} & \textbf{ISCX-IDS} \\\hline
$MDistrib_{CMS}$-R & 0.953 $\pm$ 0.001  & 0.986 $\pm$ 0.003 & \textbf{ 0.979} $\pm$ \textbf{0.01} & 0.908 $\pm$ 0.002 & 0.806 $\pm$ 0.001 \\
\hline
Spotlight & 0.89  $\pm$ 0.001 & 0.67 $\pm$ 0.002 & 0.42 $\pm$ 0.003 & 0.74 $\pm$ 0.002 & 0.63 $\pm$ 0.001 \\
\hline
MIDAS-R & 0.95  $\pm$ 0.005 & 0.984 $\pm$ 0.003 & 0.96 $\pm$ 0.01 & 0.973 $\pm$ 0.005 & 0.81 \\
\hline
PENimer & 0.87 & - & - & 0.604 & 0.51 \\
\hline
F-Fade & 0.91 $\pm$ 0.005 & 0.72 $\pm$ 0.02 & 0.61 $\pm$ 0.03 & 0.802 $\pm$ 0.001 & 0.62 $\pm$ 0.004 \\
\hline
$MDistrib_{FIS}$-R & \textbf{0.98} $\pm$ \textbf{0.0004} & \textbf{0.986} $\pm$ \textbf{0.001} & 0.96 $\pm$ 0.0005 & \textbf{0.92} $\pm$ \textbf{0.004} & \textbf{0.93} $\pm$ \textbf{0.0001} \\
\hline
\end{tabular}
}
\caption{AUC scores of the proposed methods in comparison with the edge anomaly detection baselines. The best values over a column (for a specific dataset) are shown in bold.}
\label{Table : AUC}
\end{table*}

\begin{table*}[h!]
\centering
\resizebox{1.4\columnwidth}{!}{%
\begin{tabular}{|l|c|c|c|c|c|}
\hline
Methods & \textbf{DARPA} & \textbf{CIC-DDoS} & \textbf{CIC-IDS} & \textbf{CTU-13} & \textbf{ISCX-IDS} \\\hline
$MDistrib_{CMS}$-R & 857ms  & 971ms & 741ms & 436ms & 3.7s \\
\hline
Spotlight & 129.1s & 697.6s & 209.6s & 86s & 19.5s \\
\hline
MIDAS-R & 1.4s & 2.2s & 1.1s  & 0.8s &  5.3s \\
\hline
PENimer & 5.21hrs & $\geq$ 24hrs & 10hrs & 8hrs & 1.3hrs \\
\hline
F-Fade & 317.8s  & 18.7s & 279.7s & 19.2s & 137.4s \\
\hline
$MDistrib_{FIS}$-R & 832ms & 962ms & 764.8s & 441ms & 3.1s \\
\hline
\end{tabular}%
}
\caption{Running time when detecting edge anomalies.}
\label{Table : TimeTaken}
\end{table*}

    `





\subsection{Q2. Scalability}

\begin{figure}[h!]
\centering
\includegraphics[width=1.0\linewidth, height=0.5\columnwidth]{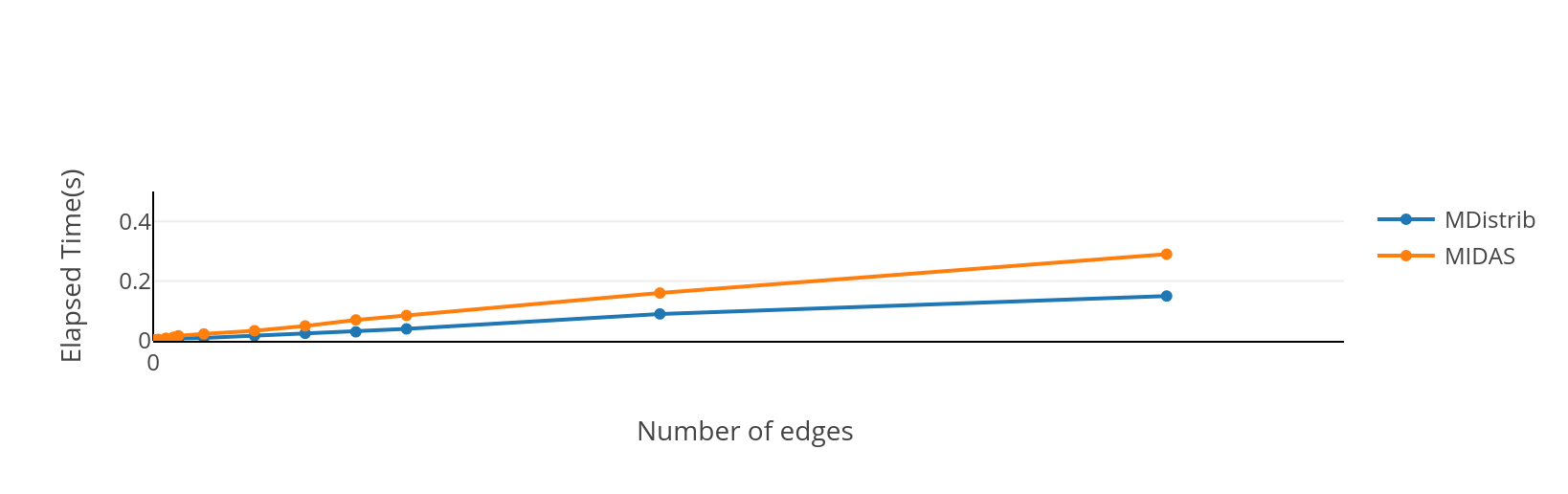}
\caption{\method w.r.t MIDAS - Scalability with respect to the number of edges}
\label{fig:midasDistScalability}
\end{figure}

In Figure \ref{fig:midasDistScalability} we showcase how via distributed processing of individual nodes and computing the sketch results in significant time improvement of execution. 
We plot the wall clock time needed to run the first $2^{12}, 2^{13}, 2^{14} \dots 2^{22}$ number of edges of the DARPA dataset for both $\method$ and $\methodmidas$.
At the same time, we showcase that the number of edges increase exponentially, $\method$ shows linearity of increase in execution time similarly as compared to $\methodmidas$ for number of partitions = 1463.

\begin{figure}[h!]
\centering
\begin{subfigure}{.8\linewidth}
  \centering
  \includegraphics[width=0.98\linewidth, height=0.5\linewidth]{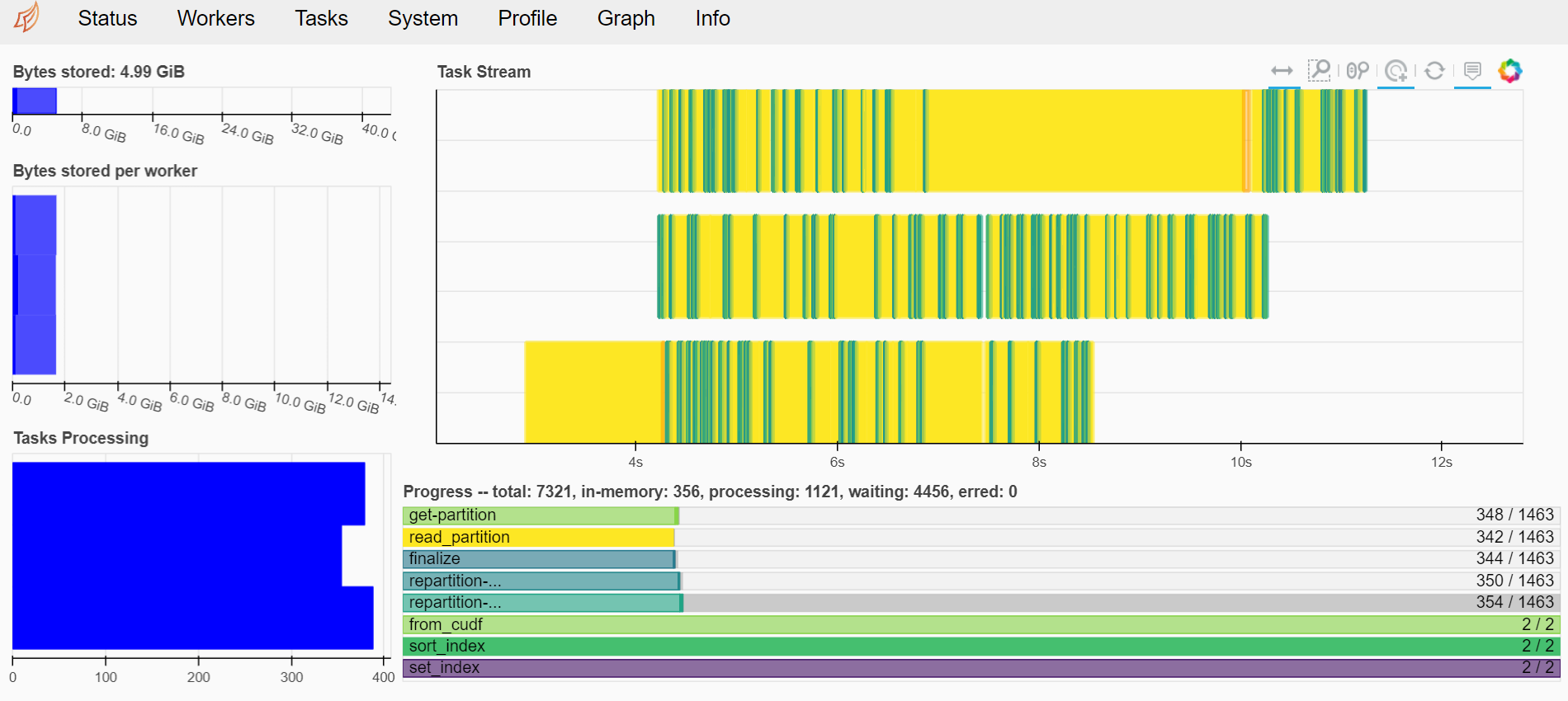}
  \caption{Saturn Cloud DashBoard Status for workers = 3}
  \label{fig:dashboard3}
\end{subfigure}%

\begin{subfigure}{.8\linewidth}
  \centering
  \includegraphics[width=.98\linewidth, height=0.45\columnwidth]{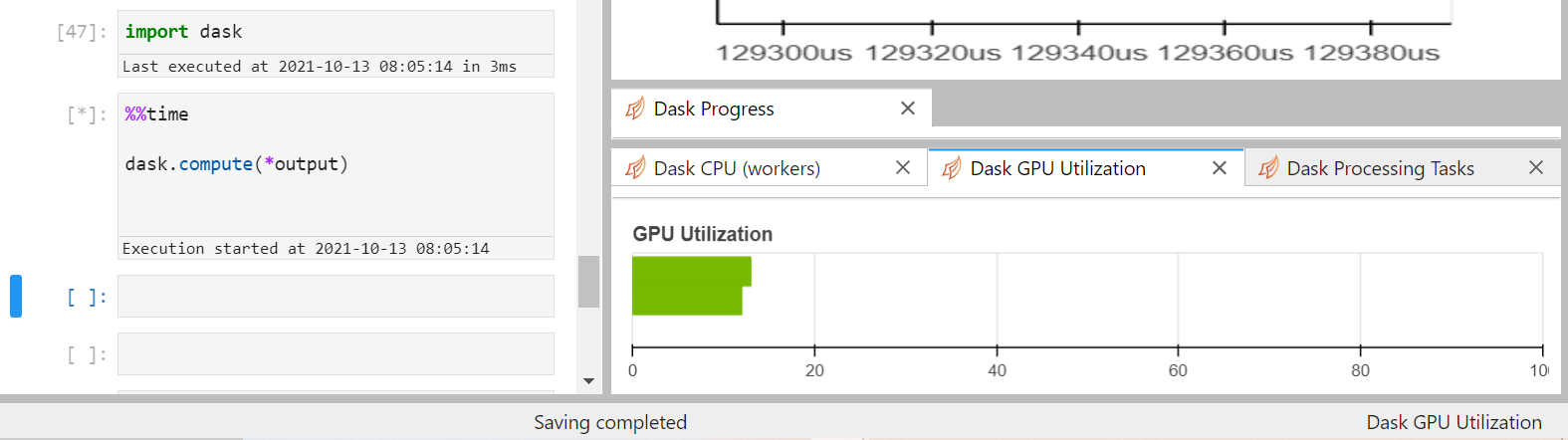}
  \caption{GPU Utilization}
  \label{fig:sub2}
\end{subfigure}

\caption{Running MDistrib in a GPU environment with Saturn Cloud}
\label{fig:prefaceComm}
\end{figure}

\begin{figure}[h!]
\centering
\begin{subfigure}{.96\linewidth}
  \centering
  \includegraphics[width=0.88\linewidth, height=0.45\linewidth]{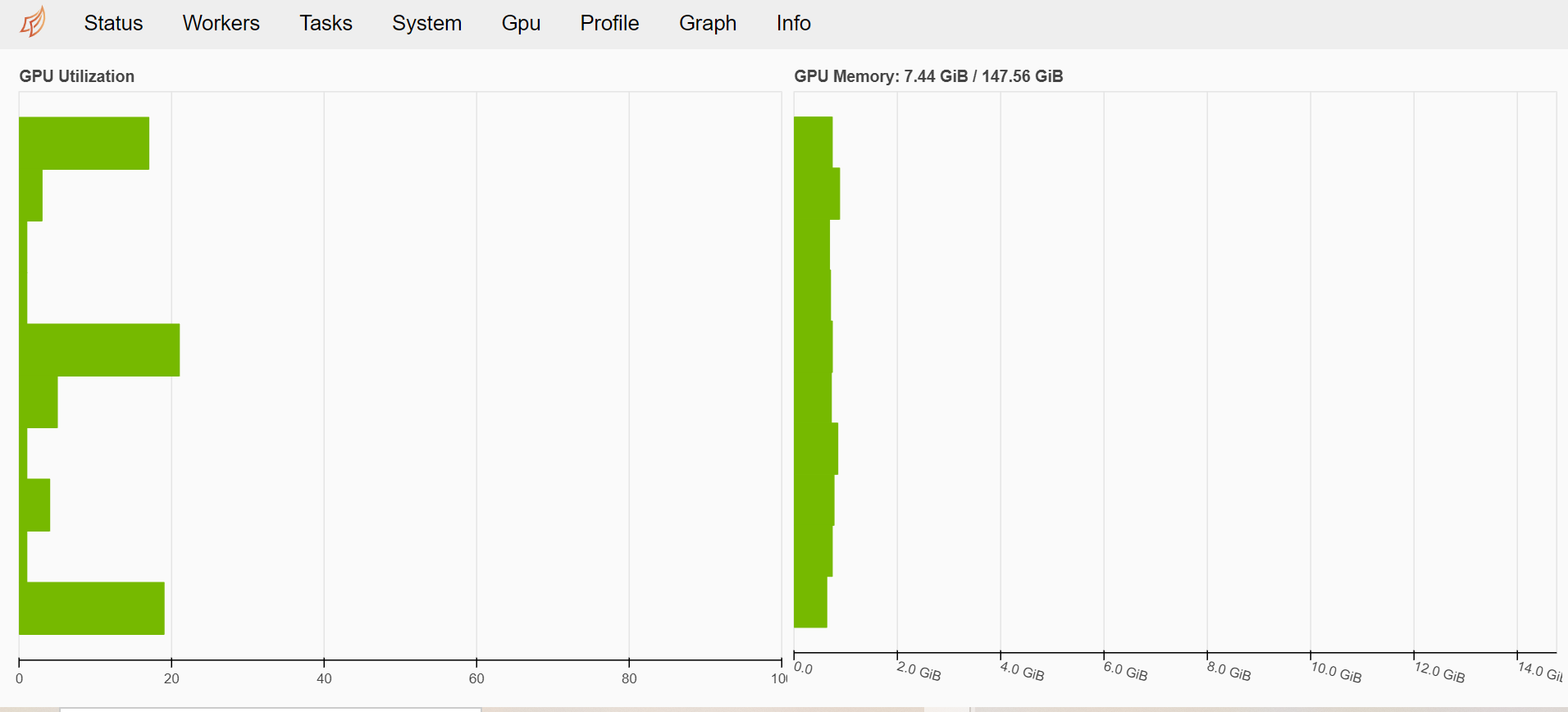}
  \caption{GPU utilization}
  \label{fig:gpuutilization10}
\end{subfigure}%

\begin{subfigure}{.97\linewidth}
  \centering
  \includegraphics[width=.94\linewidth, height=0.45\columnwidth]{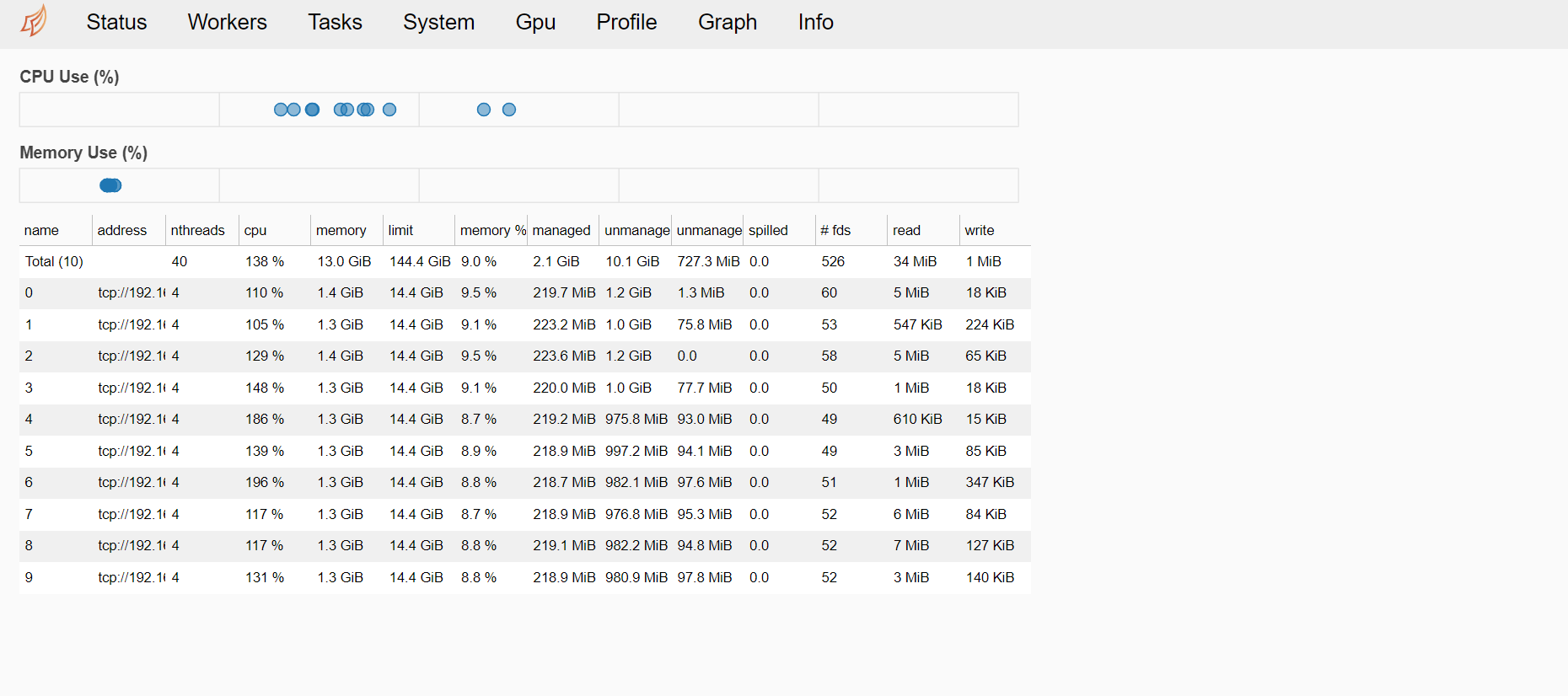}
  \caption{Resource utilization for 10 workers}
  \label{fig:resource10}
\end{subfigure}%

\caption{Results obtained by setting the number of workers to 10 in SaturnCloud.}
\label{fig:Overall worker status}
\end{figure}

\begin{figure}[h!]
  \centering
  \includegraphics[width=.84\linewidth, height=0.45\columnwidth]{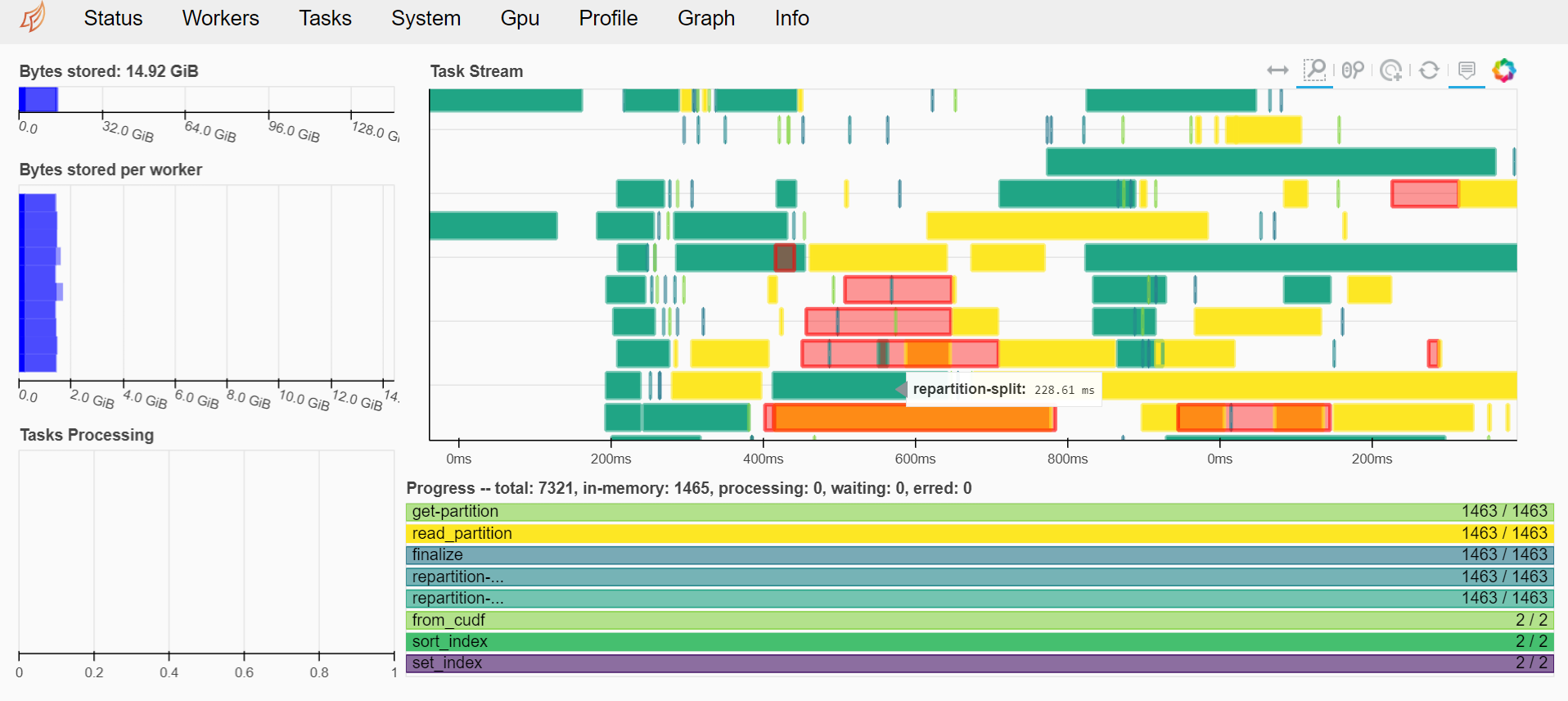}
  \caption{Task progress status}
  \label{fig:last_dashboard10}
\end{figure}

\subsection{Q3. Execution : GPU based Execution}

In addition to our proposed solution, we also extend $\method$ to a GPU setting. For this we created a dask cluster hosted on Saturn cloud\footnote{\href{https://saturncloud.io/}{https://saturncloud.io/}}  and ran our analysis on the same. We use Python's JupyterLab for running our experiments with server configuration of \textit{T4-XLarge, with 4 crores, 16GB of RAM and 1 GPU}. In addition we use cuDF\footnote{\href{https://docs.rapids.ai/api/cudf/stable/}{https://docs.rapids.ai/api/cudf/stable/}}, python's GPU dataframe library to store our data. We showcase here the experiments done for the DARPA dataset here. Using the above configurations and setting $workers = 3$, we achieved an execution time of \textbf{672ms} which is $1.37$ times faster than $\method$ in non-GPU environment.

As a separate experiment, we increased the number of workers to 10 and achieved execution time of 438ms on DARPA dataset. Further we showcase the GPU utilization and the task status for each of the 10 worker threads. Fig.~\ref{fig:gpuutilization10} shows the GPU utilization for $\method$ using partitions = 1463 and number of workers = 10.

Figure.\ref{fig:dashboard3} and Figure.\ref{fig:last_dashboard10} shows the task status for worker numbers =3 and 10 respectively. Each of the task status depicts the number of underlying processes involved like reading data from a partition and processing it and writing it. Figure \ref{fig:resource10} thus depicts the resources like number of threads or cpu cycles being utilized during the program execution by each of the 10 workers.

\section{Conclusion}

In this paper, we proposed \method which uses a distributed strategy to detect anomalous edges in an offline setting. We further showcase that through our distributed sketch architecture, we achieve a better false positive theoretical guarantee and as well as better performance in terms of running time as compared to state of the art methods for large graphs (billion number of edges). We also extend our method to GPU enabled environment to showcase the applicability of our system in high performance computing environments. Lastly we showcased a generalized model of our framework that can use any frequency based sketch and achieve good accuracy results in detecting suspicious anomalies.

\bibliographystyle{plain}
\bibliography{sample}

\end{document}